\documentclass{article}

\usepackage{booktabs} 

\usepackage[ruled]{algorithm2e} 

\usepackage{microtype}

\usepackage{graphicx}

\usepackage{amsmath}
\usepackage{amsthm}
\usepackage{amsxtra}
\usepackage{amssymb}

\usepackage{verbatim} 
\usepackage{hyperref}
\usepackage[utf8]{inputenc}

\newcommand{\C}{\mathcal{C}}
\newcommand{\DEC}{\operatorname{DEC}}
\newcommand{\F}{\mathbb{F}}
\newcommand{\Exp}[2][]{\mathbb{E} \ifx\\#1\else_{#1}\fi\left[ #2 \right]}
\newcommand{\length}{\operatorname{length}}
\renewcommand{\O}{\mathcal{O}}
\renewcommand{\P}{\mathcal{P}}

\newcommand{\RS}{\mathcal{RS}}
\newcommand{\spane}{\operatorname{span}}
\newcommand{\I}{\mathcal{I}}

\newcommand{\M}{\mathcal{M}}

\newtheorem{definition}{Definition}
\newtheorem{lemma}[definition]{Lemma}

\newtheorem{corollary}[definition]{Corollary}
\newtheorem{theorem}[definition]{Theorem}

\newtheorem{claim}[definition]{Claim}

\date{}
\title{New Amortized Cell-Probe Lower Bounds for Dynamic Problems}
\author{
  Sayan Bhattacharya\\
  University of Warwick\\
  Warwick, United Kingdom
  \and
  Monika Henzinger \\
  University of Vienna \\
  Faculty of Computer Science \\
  Vienna, Austria
  \and 
  Stefan Neumann \\
  University of Vienna \\
  Faculty of Computer Science \\
  Vienna, Austria
}

\begin{document}

\maketitle

\begin{abstract}
We build upon the recent papers by Weinstein and Yu~\cite{weinstein2016amortized},
Larsen~\cite{larsen2012higher}, and Clifford et al.~\cite{clifford2015new} to present a general framework that
gives {\em amortized} lower bounds on the update and query times of dynamic data structures.
Using our framework, we present two concrete results.
\begin{enumerate}
\item For the dynamic polynomial evaluation problem, where the polynomial is defined over a finite field of size $n^{1+\Omega(1)}$ and has degree $n$, any dynamic data structure must either have an {\em amortized} update time of $\Omega((\lg n/\lg \lg n)^2)$ or an {\em amortized} query time of $\Omega((\lg n/\lg \lg n)^2)$. 
\item For the dynamic online matrix vector multiplication problem, where we get an $n \times n$ matrix whose entires are drawn from a finite field of size $n^{\Theta(1)}$, any dynamic data structure must either have an {\em amortized} update time of $\Omega((\lg n/\lg \lg n)^2)$ or an {\em amortized} query time of $\Omega(n \cdot (\lg n/\lg \lg n)^2)$. 
\end{enumerate}
For these two problems, the previous works by Larsen~\cite{larsen2012higher} and Clifford et al.\cite{clifford2015new} gave
the same lower bounds, but only for {\em worst case} update and query times. Our bounds match the highest unconditional lower bounds
known till date for any dynamic  problem in the cell-probe model.
\end{abstract}

\smallskip

\thanks{The research leading to these results has received funding from the
	European Research Council under the European Union’s Seventh Framework
		Programme (FP/2007-2013) / ERC Grant Agreement no.\ 340506.
	The third author gratefully acknowledges the financial support of the
	Doctoral Programme ``Vienna Graduate School on Computational Optimization''
		which is funded by Austrian Science Fund  (FWF, project no. W1260-N35).
}

\section{Introduction}

In an abstract \emph{dynamic problem}, we want  a data structure that supports two types of \emph{input} operations:
 \emph{updates} and  \emph{queries}.  The time taken  to handle an update (resp. query) operation is known as the {\em update time} (resp. {\em query time}). The main goal in this field is to design data structures  with small update and query times for fundamental dynamic problems.
 
\noindent {\bf The Cell-Probe Model.}
The focus of this paper is on proving {\em unconditional} lower bounds on the update and query times of dynamic data structures.\footnote{In contrast, the recent papers~\cite{AbboudW14,OMV} prove {\em conditional} lower bounds assuming  SETH and OMv conjectures.}
All such known lower bounds work in the \emph{cell-probe} model of computation~\cite{Yao}. 
In this model, the memory is organized into a set of memory {\em cells}. Each
cell stores $w = \Theta(\lg n)$ bits of information. For most dynamic problems,
specifying an input to the data structure requires $\Theta(\lg n)$ bits because
the number of inputs is polynomial in $n$. Moreover, usually
the total number of cells in the memory of the data structure is polynomially bounded by $n$.
Thus, it takes $w$ bits to specify the address of a cell.
Whenever a cell is read or written, we say that it is \emph{probed}.
The data structure supports an input operation by probing cells in the memory.
If the input operation is a query, then the data structure also outputs an answer to that query after probing cells.
The time taken to handle an input operation is measured in terms of the number of cell-probes made by the data structure.
Any other computation, such as computation inside the CPU, comes free of cost.
Intuitively, this model captures the {\em communication cost} between the CPU and the memory, and is, thus, very suitable for information theoretic arguments.

\noindent {\bf Previous Work.}
Proving  cell-probe lower bounds turns out to be technically very challenging.
One of the first major papers in this area was by Fredman and
Saks~\cite{fredman1989cell}, who introduced the \emph{chronogram
technique}~\cite{MarkedAncestor}.  Here, we construct a random sequence of
$\text{poly}(n)$ updates followed by one random query.  Going backward in
time, the sequence of updates are partitioned into $\Theta(\lg n/\lg \lg n)$
many epochs whose sizes keep increasing exponentially. To be more specific,
for $i \geq 1$, epoch $i$ consists of  $\delta^i$ consecutive updates, where
$\delta := \text{poly} \lg n$. We then {\em mark} each memory cell by the
(unique) epoch in which it was last probed. Finally, we show that if the
worst case update time is $\text{poly} \log n$, then the random query at the
end  must probe at least $\Omega(t)$ cells marked by each epoch, where $t$
is a parameter. Summing over all the epochs, this gives a lower bound of
$\Omega(t \cdot \lg n/\lg \lg n)$ on the worst case query time.
By interleaving queries and updates and the use of a counting argument,
they show how their worst case lower bound implies an amortized lower
bound.

An exciting development in this field has been the {\em cell-sampling}
technique introduced by Panigrahy et al.~\cite{panigrahy10lower} for static data
structures.
Larsen~\cite{larsen2012cell,larsen2012higher} used this technique to provide new
lower bounds for dynamic data structures. He showed how to prove a lower bound of
$t = \Omega(\lg n/\lg \lg n)$ for a {\em single epoch} in the chronogram.
Summing over all the epochs, this gives a lower bound of $\Omega((\lg n/\lg \lg n)^2)$ on the worst
case query time.  At its core, the cell-sampling technique is an encoding
argument. It says that if the query time were small, then  there  must exist a
very small subset of memory cells (say, $\C^*$) from which one can infer the
answers to a large number of queries.  Next, it exploits the fact that for {\em
some} dynamic problems, the answers to a sufficiently large number of
queries completely determine the past updates.\footnote{Intuitively, these
	dynamic problems have the property that the vector of answers is
	$\approx n$-wise independent over a random sequence of updates. That is,
	\emph{any} $n$ answers to queries (essentially) determine the input
	sequence.}
For such problems, therefore,
if the query time were too small, then we could potentially give an encoding
of the past updates by specifying the addresses and contents of the cells in
$\C^*$. Since the size of $\C^*$ is very small, this encoding would use
fewer bits than the entropy of the past updates. As this  leads to a
contradiction, we are left with no other choice but to conclude that the
query time must be large.

For other major results, see
~\cite{patrascu2006logarithmic,patrascu2011dont,Yu16}.  None of them, however,
can give $\omega(\lg n)$ cell-probe lower bounds.  For only three dynamic
problems $\omega(\lg n)$ lower bounds are known, and they all follow from
the cell-sampling technique: (1) 2d range counting, (2) polynomial
evaluation and (3) OMv. Initially, all these lower bounds were for worst
case update and query times.

\noindent {\bf Amortized lower bounds.} For many problems there is a large gap
between worst case and amortized lower bounds.  Furthermore, from the
perspective of a practitioner, amortized data structures are in many
applications as useful as worst case data structures. Very recently,
Weinstein and Yu~\cite{weinstein2016amortized} showed how to make Larsen's lower bound
for 2d range counting work in the amortized setting. Their lower bound also
applies to data structures with very high error probability. But it remained
open whether one can get such amortized lower bounds for the remaining two
problems that are known to be solvable via the cell-sampling technique: (1)
polynomial evaluation and (2) OMv. We resolve this question in the affirmative.
In addition, we show a generic template for proving such amortized lower bounds.
Specifically, in Definition~\ref{def:well:behave}, we introduce the notion of a
{\em well-behaved} input sequence for a dynamic problem. We prove that if one
can show the existence of a well-behaved input sequence, then the corresponding
dynamic problem admits the desired amortized lower bound.

\noindent {\bf Remark.} Fredman and Saks~\cite{fredman1989cell} mention an
extension of the chronogram technique where the input sequence consists of
multiple queries perfectly interleaved with the updates. It might be possible to
get an alternate proof of our results using this approach.  However, we believe
that the framework we present is going to be useful for the cleaner exposition
of future cell-probe lower bounds.

\subsection{Our Results}
\label{Sec:GeneralFramework}

Recall that an {\em input} means either an {\em update} or a {\em query}.
Throughout this paper, we let $\O$ denote a (random) sequence of
$\Theta(n^{\alpha})$ inputs for the dynamic problem under consideration, where
$\alpha > 0$ is some constant.  Our goal is to give a lower bound of the
expected number of cell-probes that any data structure has to make while
processing the input sequence $\O$.  
By  abusing the  notation, we use the symbol  $X \subseteq \O$  to denote a
contiguous {\em interval}  of inputs in $\O$. We  let $|X|$ denote the size of
the interval $X$, which refers to the number of inputs in $X$. We define $P(X)$
to be the set of cells probed by a data structure while processing the inputs in
$X$.  We say that two intervals $X, Y \subseteq \O$ are {\em consecutive} if the
first input in $Y$ appears immediately after the last input in $X$.    For any
two consecutive intervals $X, Y \subseteq \O$, the counter $C(X, Y)$ denotes the
number of times the following event occurs:
While answering a {\em query} in $Y$, the data structure reads a cell
that was updated in $X$, and is probed in $Y$, but the cell was not probed
earlier in $Y$.\footnote{Note that the
	authors in~\cite{weinstein2016amortized} do not use this counter. Instead
	they consider the set of cells $P(X) \cap P(Y)$. In contrast, we only
	consider the cells probed in $Y$ for answering {\em queries}.}
The reason for this definition of $C(X,Y)$ is to avoid double counting cells
that are written during $X$ and then probed multiple times during $Y$.

In Definition~\ref{def:well:behave}, we introduce the concept of a
``well-behaved'' input sequence, which implicitly captures the main idea behind
the dynamic lower bounds  obtained by the chronogram method.

\begin{definition}
\label{def:communicationLowerBound}
\label{prop:main}
\label{def:well:behave}
Let $\kappa \in (0,1)$ and $c, \alpha, \beta > 0$ be constants, and define $\gamma := \lg^{\beta} n$. Fix a dynamic problem $\P$, and consider a (random) input sequence $\O$ of size $\Theta(n^{\alpha})$ for this problem. Such an input sequence $\O$  is called ``well-behaved''  iff  every dynamic data structure  for   problem $\P$ satisfies the following property while processing the  updates and queries in $\O$. 
  \begin{itemize}
 \item For every pair of consecutive intervals  $X, Y \subseteq \mathcal{O}$ with
  $|X| \simeq \gamma \cdot |Y|$ and $|Y| \geq |\O|^{\kappa}$, at least one of these  three conditions  is violated:
  \begin{enumerate}
		\item  $\Exp{ |P(X)| } < |X| \cdot \lg^2 n$,
		\item $\Exp{ |P(Y)| } < |Y| \cdot \lg^2 n$, and
		\item $\Exp{ C(X,Y) } < \frac{1}{200 (c + 1002)} \cdot |Y| \cdot \frac{\lg n}{\lg \lg n}$.
	\end{enumerate}
  \end{itemize}
\end{definition}

We now explain this in a bit more details.
Suppose that the input sequence $\O$ consists of a batch
of random updates followed by one random query. Build a chronogram on top of
this input sequence $\O$ and consider any {\em sufficiently large} epoch $i$ in
this chronogram, ensuring that the number of inputs appearing after epoch $i$ is
$\Omega(|\O|^{\kappa})$ for some small constant $\kappa \in (0, 1)$.  The total
number of such large epochs is still $\Omega(\lg n/\lg\lg n)$. Hence,  we do not
incur any asymptotic loss in the derived lower bound if we focus only on these
large epochs. Let $X$ denote the sequence of inputs in epoch $i$, and let $Y$
denote the sequence of all the inputs in $\O$ that appear after $X$.  As we go
back in time, the sizes of the epochs in a chronogram increase exponentially in
some $\gamma := \text{poly} \lg n$ factor. Thus, we have $|X| \simeq \gamma
\cdot |Y|$ and $|Y| \geq |\O|^{\kappa}$. Intuitively, a chronogram based lower bound
basically proves the following statement.

\begin{claim}[Informal]
\label{cl:ch}
If the update time of the data structure is some small $\text{poly} \lg n$, say $\lg^2
n$,  then to answer the random query at the end of the input sequence, the data
structure must probe at least $\Omega(t)$ cells that were probed in $X$ but not
probed before in $Y$, for some parameter $t$.
\end{claim}

To notice the similarities between Claim~\ref{cl:ch} and
Definiton~\ref{def:well:behave}, interpret Definition~\ref{def:well:behave} as
follows: A well-behaved input sequence $\O$ is such that if any data structure
satisfies conditions (1) and (2), then it must violate condition (3). Seen in
this light, conditions (1) and (2) in Definition~\ref{def:well:behave} are
analogous to the statement that the update time of the data structure is at most
$\lg^2 n$. Similarly, the assertion that condition (3) must be violated becomes
analogous to the statement that while answering the random query at the end of
the input sequence, the data structure must read  many cells that were probed in
$X$ but not probed before in $Y$.

Definition~\ref{def:well:behave}, however, is much more general than
Claim~\ref{cl:ch}. For example,  typically a well-behaved sequence will
intersperse the queries with the updates, instead of having only one query at
the end of all the updates. Furthermore, as opposed to the classical chronogram
method, the input sequence $\O$ need not end with the interval $Y$. These
important distinctions between Definition~\ref{def:well:behave} and
Claim~\ref{cl:ch} help us derive {\em amortized} lower bounds using our
framework.

Theorem~\ref{Thm:GeneralFramework} shows how the existence of a well-behaved input sequence implies a cell-probe lower bound for the dynamic problem under consideration. Its proof appears in Section~\ref{sec:Thm:GeneralFramework}. We  use Theorem~\ref{Thm:GeneralFramework} to derive {\em amortized} cell-probe lower bounds for two concrete problems.

\begin{theorem}
\label{Thm:GeneralFramework}
	Let $\O$ be a well-behaved (random) input sequence for a dynamic  problem $\P$ as per  Definition~\ref{def:well:behave}.
	Then any dynamic data structure for  problem $\P$  needs to probe
	at least $\Omega(|\O| \cdot (\lg n / \lg \lg n)^2)$ cells in expectation while processing the input sequence $\O$.
\end{theorem}

We note that in the upcoming proofs of our lower bounds it is crucial that a
well-behaved input sequence $\O$ additionally satisfies the following two properties:
(1) $\O$ interleaves updates and queries. This is necessary because, for
example, if we only had a single query at the end of $\O$, then an
\emph{amortized} data structure could just batch all updates together and solve
the static version of the problem. The static version of the problem might,
however, allow for faster algorithms than the dynamic version of the problem.
(2) The updates and queries in $\O$ are independent operations. This is crucial
for the chronogram argument and cell-sampling encoding proofs to go through.

\smallskip
\noindent {\bf Our Result on Dynamic Online Matrix Vector Multiplication (OMv).}

\smallskip
\noindent Consider an $n \times n$ matrix $M$ over a finite field $\F$ of size
$|\F| = n^{\Theta(1)}$. All the entries in this matrix are set to zero in the
beginning.  Subsequently, the data structure should be able to handle any
sequence of two types of {\em operations}:
\begin{itemize}
	\item {\sc Update} $(i,j,x) \in \{1,\dots,n\}^2 \times \F$: Set the entry $(i,j)$ of $M$, denoted as $M_{ij}$, to $x \in \F$.
	\item {\sc Query} $v \in \F^n$: Return  the matrix vector product $M \cdot v$.
\end{itemize}
\noindent Note that it requires $\Theta(\lg n)$ bits to specify an update, and
$\Theta(n \cdot \lg n)$ bits to specify the answer to a query. Since each cell
contains $w = \Theta(\lg n)$ bits, it is trivial to show a lower bound of
$\Omega(1)$ on the update time, and a lower bound of $\Omega(n)$ on the query
time.  Our main result on the dynamic OMv problem is summarized in
Theorem~\ref{Thm:DynamicOMv}. To prove Theorem~\ref{Thm:DynamicOMv}, we  adapt
the approach of Clifford et al.~\cite{clifford2015new} into our setting. See
Section~\ref{sec:omv} for the detailed proof.

\begin{theorem}
\label{Thm:DynamicOMv}
	For the dynamic OMv problem, there exists a well-behaved random input sequence $\O$ consisting of $n^2$ updates and $n$ queries.
\end{theorem}

\begin{corollary}
	A cell-probe data structure for the dynamic OMv problem
	must have a total update and query time of $\Omega( n^2 (\lg n / \lg \lg n)^2)$
	over a sequence of $n^2$ updates and $n$ queries.
\end{corollary}

\begin{proof}
The input sequence $\O$ as described in Theorem~\ref{Thm:DynamicOMv} consists of $n^2$ updates and $n$ queries. Since $\O$ is well-behaved,  the total time taken to process the inputs in $\O$ is $\Omega(n^2 \cdot (\lg n/\lg \lg n)^2)$ in the cell-probe model (see Theorem~\ref{Thm:GeneralFramework}). Hence, either $\Omega(n^2 \cdot (\lg n/\lg \lg n)^2)$ many cells are probed while processing the $n^2$ updates, or $\Omega(n^2 \cdot (\lg n/\lg \lg n)^2)$ many cells are probed while processing the $n$ queries.
\end{proof}

\smallskip
\noindent {\bf Our Result on Dynamic Polynomial Evaluation.}

\smallskip
\noindent
Here, the data structure gets a polynomial  of degree $n$ over a finite field
$\F$ of size $n^{1+\Omega(1)}$. It is specified as $(x-r_1) \cdot (x-r_2) \cdots
(x-r_n)$, where $r_i$ is the $i^{th}$ root. Subsequently, the data structure
should be able to handle  the following two types of {\em operations}:
\begin{itemize}
\item {\sc Update($i, z$):} Set the $i^{th}$ root to $z \in \F$, that is, set $r_i := z$.  
\item {\sc Query($x$):} Evaluate the value of the polynomial at $x \in \F$. 
\end{itemize}
\noindent Both an update and a query here can be specified using $\Theta(\lg n)$
bits. Since each cell contains $w = \Theta(\lg n)$ bits, it is trivial to show a
lower bound of $\Omega(1)$ on the update or the query time.  Our main result on
the dynamic polynomial evaluation problem is summarized in
Theorem~\ref{thm:main:poly}, which we prove in Section~\ref{app:Sec:PolynomialEvaluation}.

\begin{theorem}
\label{thm:main:poly}
	For the dynamic polynomial evaluation problem, there exists a well-behaved
	random input sequence $\O$ consisting of $n$ updates and $n$ queries. 
\end{theorem}

\begin{corollary}
\label{thm:cor:poly}
	A cell-probe data structure for the dynamic polynomial evaluation
	problem must have a total update and query time of $\Omega(n (\lg n/\lg \lg n)^2)$
	over a sequence of $n$ updates and $n$ queries.
\end{corollary}

\begin{proof}(Sketch)
Follows from Theorem~\ref{Thm:GeneralFramework} and Theorem~\ref{thm:main:poly}.
\end{proof}

\section{Proof of Theorem~\ref{Thm:GeneralFramework}}
\label{sec:Thm:GeneralFramework}

For every interval $I \subseteq \mathcal{O}$
we define  $\DEC(I) := (I_A, I_B)$, where $I_B \subseteq I$ are the last $\lfloor |I| / \gamma \rfloor$ inputs in the interval $I$, 
and $I_A$ are  the first $|I| -\lfloor |I| / \gamma \rfloor$ inputs  in the interval $I$. Note that $I_A$ and $I_B$ partition $I$, and we have  $|I_A| \simeq \gamma \cdot |I_B|$.

Following the framework of Weinstein and Yu~\cite{weinstein2016amortized} we  construct a ``hierarchy'' $\I$, which can be thought of as a rooted binary tree that is  built on top of the input sequence $\O$. Every  node in this tree corresponds to an interval $I \subseteq \O$, and the root corresponds to the entire sequence $\O$. Every non-leaf node $I \subseteq \O$ has two children $I_A$ and $I_B$ such that $\DEC(I) = (I_A, I_B)$. We ensure that every node $I \subseteq \O$ in this tree has size $|I| \geq |\O|^{\kappa}$. In other words, as we move down a path from the root, the intervals corresponding to the nodes on this path keep getting smaller and smaller in size. Consider the first (closest to the root) node $I' \subseteq \O$ on this path whose size is less than $\gamma \cdot |\O|^{\kappa}$. If the node $I'$ had two children $I'_A$ and $I'_B$ such that $\DEC(I') = (I'_A, I'_B)$, then the size of $I'_B$ would be less than $|\O|^{\kappa}$. In order to rule out this possibility, such a node $I'$ becomes a leaf in the tree. 

An interesting corollary of this construction is as follows. Consider any two nodes $I_A, I_B \subseteq \O$ in this tree that are ``siblings'' of each other, meaning that they share the same parent node $\I \subseteq \O$ and $\DEC(I) = (I_A, I_B)$. Then the two intervals $I_A, I_B$ are consecutive, $|I_A| \simeq \gamma \cdot |I_B|$, and $|I_B| \geq |\O|^{\kappa}$. Hence, the property of a well-behaved input sequence $\O$ as stated in Definition~\ref{def:well:behave} will apply to these two intervals $I_A$ and $I_B$.

By convention,  the root of this tree is at level $0$, and the level of a child node is  one more than that of its parent. With this convention in mind, for every integer $i \geq 1$ let $\I_i$  be the collection of ordered pairs of siblings $(I_A, I_B)$ that constitute the level $i$ of the hierarchy $\I$.
We further  define $i_{\max} := 0.1 \gamma \lg |\O| / \lg \gamma$. Lemma~\ref{lem:sizeOfIBs} lower bounds the total size of the $I_B$ intervals at any level $i \leq i_{max}$ of this hierarchy. Theorem~\ref{Thm:GeneralFramework} follows from Lemma~\ref{cl:hierarchyLevel} and~\ref{cl:noEfficientDatastructure}.

\begin{lemma}[\cite{weinstein2016amortized}, Claim~4]
\label{lem:sizeOfIBs}
  For each level $i \leq i_{\max}$, we have $
	\sum_{(I_A, I_B) \in \mathcal{I}_i} |I_B| \geq |\O|/(2\gamma)$.
\end{lemma}

\begin{lemma}
\label{cl:hierarchyLevel}
	For every  level $i \leq i_{\max}$, at least one of these  three conditions is violated:
	\begin{enumerate}
		\item $\sum_{(I_A,I_B) \in \mathcal{I}_i} \Exp{ |P(I_A)| } < \frac{|\O| \cdot \lg^2 n}{8}$.
		\item $\sum_{(I_A,I_B) \in \mathcal{I}_i} \Exp{ |P(I_B)| } < \frac{|\O| \cdot \lg^2 n}{8\gamma}$.
		\item $\sum_{(I_A,I_B) \in \mathcal{I}_i} \Exp{ C(I_A,I_B) } <
		\frac{1}{1200 (c + 1002)} \cdot \frac{|\O|}{\gamma} \cdot \frac{\lg n}{\lg \lg n}$.
	\end{enumerate}
\end{lemma}

\begin{proof}
The input sequence $\O$ is well-behaved as per Definition~\ref{def:well:behave}. 
For $j \in \{1,2,3\}$, define
\begin{align*}
F_j := \{ (I_A,I_B) \in \mathcal{I}_i :
			(I_A,I_B) \text{ violates condition $j$ of Definition~\ref{def:communicationLowerBound}}\}.
\end{align*}
We also define $\length(F_j) := \sum_{(I_A,I_B) \in F_j} |I_B|$. As per
Definition~\ref{def:well:behave}, each ordered pair $(I_A, I_B) \in \I_i$
belongs to at least one of the $F_j$'s. Furthermore, by Lemma~\ref{lem:sizeOfIBs}
we have $\sum_{(I_A,I_B) \in \mathcal{I}_i} |I_B| \geq \frac{|\O|}{2\gamma}$.
Thus, we infer that $\sum_{j = 1}^3 \length(F_j) \geq \frac{|\O|}{2\gamma}$, and
hence there must be some $j \in \{1, 2, 3\}$ for which
$\length(F_j) \geq \frac{|\O|}{6 \gamma}$.
We now fork into three cases, and show that in each
case one of the three conditions stated in
Lemma~\ref{cl:hierarchyLevel} gets violated. 

\smallskip
\noindent {\bf Case 1: $\length(F_1) \geq \frac{|\O|}{6\gamma}$.}
In this case, we can derive that:
\begin{align*}
	\sum_{(I_A,I_B) \in \mathcal{I}_i} \Exp{ |P(I_A)| }
	&\geq \sum_{(I_A,I_B) \in F_1} \Exp{ |P(I_A)| } \\
	&\geq \sum_{(I_A,I_B) \in F_1} |I_A| \cdot \lg^2 n  \\
	&\geq \sum_{(I_A,I_B) \in F_1} |I_B| \cdot \gamma \cdot \lg^2 n \\
	&= \length(F_1) \cdot \gamma \cdot \log^2 n \geq \frac{|\O| \cdot \lg^2 n}{6}.
\end{align*}
The second inequality follows from the definition of the set $F_1$.
The third inequality holds since $|I_A| \simeq |I_B| \cdot \gamma$.
Hence, the first condition in Lemma~\ref{cl:hierarchyLevel} is violated. 

\smallskip
\noindent {\bf Case 2: $\length(F_2) \geq \frac{|\O|}{6\gamma}$.}
In this case, we can derive that:
\begin{align*}
	\sum_{(I_A,I_B) \in \mathcal{I}_i} \Exp{ |P(I_B)| }
	&\geq \sum_{(I_A,I_B) \in F_2} \Exp{ |P(I_B)| } \\
	&\geq \sum_{(I_A,I_B) \in F_2} |I_B| \cdot \lg^2 n \\
	&\geq \length(F_2) \cdot \log^2 n \\
	&\geq \frac{|\O| \cdot \lg^2 n}{6 \gamma}.
\end{align*}
The second inequality follows from the definition of the set $F_2$. We conclude that in this case the second condition in Lemma~\ref{cl:hierarchyLevel} is violated.  

\smallskip
\noindent {\bf Case 3: $\length(F_3) \geq \frac{|\O|}{6\gamma}$.}
In this case, we can derive that:
\begin{align*}
	\sum_{(I_A,I_B) \in \mathcal{I}_i} \Exp{ C(I_A, I_B) }
	&\geq \sum_{(I_A,I_B) \in F_3} \Exp{ C(I_A, I_B) } \\
	&\geq \sum_{(I_A,I_B) \in F_3} \frac{1}{200 (c+1002)} \cdot |I_B| \cdot \frac{\lg n}{\lg \lg n}  \\
	&\geq \frac{1}{1200 (c+1002)} \cdot \frac{|\O|}{\gamma} \cdot \frac{\lg n}{\lg \lg n}.
\end{align*}
The inequalities follow from the definition of $F_3$.  The computation shows
that the third condition in Lemma~\ref{cl:hierarchyLevel} is violated. 
\end{proof}

\begin{lemma}
\label{cl:noEfficientDatastructure}
  Suppose that a data structure probes $o(|\O| \cdot (\lg n / \lg \lg n)^2)$ cells in expectation while processing the updates and queries in $\O$. Then there exists a level $i \leq i_{\max}$ such that all of the following conditions are satisfied:
	\begin{enumerate}
		\item $\sum_{(I_A,I_B) \in \mathcal{I}_i} \Exp{ |P(I_A)| } < \frac{|\O| \cdot \lg^2 n}{8}$.
		\item $\sum_{(I_A,I_B) \in \mathcal{I}_i} \Exp{ |P(I_B)| } < \frac{|\O| \cdot \lg^2 n}{8\gamma}$.
		\item $\sum_{(I_A,I_B) \in \mathcal{I}_i} \Exp{ C(I_A, I_B) } < 
			\frac{1}{1200 (c + 1002)} \cdot \frac{|\O|}{\gamma} \cdot \frac{\lg n}{\lg \lg n}$.
	\end{enumerate}
\end{lemma}

We devote the rest of this section to the proof of Lemma~\ref{cl:noEfficientDatastructure}.
For $j \in \{1, 2, 3\}$,  say that a level $i \in \{1, \ldots, i_{max}\}$ in the hierarchy $\I$ is of ``type $j$''  iff it violates condition $j$ in Lemma~\ref{cl:noEfficientDatastructure}. In Claim~\ref{cl:new:1}, we show that there is no level of type $1$ in $\{1, \ldots, i_{max}\}$.  Claims~\ref{cl:new:2},~\ref{cl:new:3} state that for each $j \in \{2, 3\}$, the number of levels in $\{1, \ldots, i_{max}\}$ that are of type $j$ is  less than $i_{max}/2$. Hence, there exists some level $i \in \{1, \ldots, i_{max}\}$ that is not of type $1, 2$ or $3$. By definition, such a level $i$ satisfies all the three conditions  in Lemma~\ref{cl:noEfficientDatastructure}.

\begin{claim}
\label{cl:new:1}
There is no level $i \in \{1, \ldots, i_{max}\}$ that is of type $1$. 
\end{claim}

\begin{proof}
Let there be a level $i \leq i_{max}$  of type $1$, then
$\sum_{(I_A,I_B) \in \mathcal{I}_i} \Exp{ |P(I_A)| } \geq \frac{|\O| \cdot \lg^2 n}{8}$.
The sum $\sum_{(I_A,I_B) \in \mathcal{I}_i} |P(I_A)|$ is a lower bound on the total number of cells probed by the data structure while processing the updates and queries in $\O$. Hence, this inequality implies that the data structure probes at least $\Omega(|\O| \cdot \log^2 n)$ cells in expectation while processing the input sequence $\O$. This contradicts the  assumption specified in Lemma~\ref{cl:noEfficientDatastructure}.
\end{proof}

\begin{claim}
\label{cl:new:2}
The number of type $2$ levels in $\{1, \ldots, i_{max}\}$ is strictly less than $i_{max}/2$. 
\end{claim}

\begin{proof}
Consider any  element $x$ in the sequence of updates and queries $\O$ and any level $i \in \{1, \ldots, i_{max}\}$. If there exists an ordered pair $(I_A, I_B) \in \I_i$ such that $x \in I_B$, then we say that $x$ ``appears'' in level $i$ and that the ``window'' of $x$ at level $i$ is equal to $|I_B|$. 

Fix any element $x$ in $\O$ and  scan through the levels $\{1, \ldots,
i_{max}\}$ in the hierarchy $\I$ in increasing order.  Clearly, the window
of $x$ at any level it appears in is at most $|\O|$. Further,  every time
the element $x$ appears in a level during this scan, its window shrinks by
at least a factor of $\gamma$. This property holds since  $|I_B| \simeq
(1/\gamma) \cdot |I|$ whenever we have $\DEC(I) = (I_A, I_B)$. Thus,  any
element in $\O$ can appear in at most $\log_{\gamma} |\O|$ levels.
Accordingly, from a simple counting argument, it follows that the number of
cell-probes made by the data structure while processing the input sequence
$\O$ is at least
\begin{align*}
	\Gamma := (1/\log_{\gamma} |\O|) \cdot \sum_{i =1}^{i_{max}} \sum_{(I_A,I_B) \in \mathcal{I}_i} |P(I_B)|.
\end{align*}
If the number of type $2$ levels in $\{1, \ldots, i_{max}\}$ were at least
$i_{max}/2$, then we would get:
\begin{align*}
\Exp{\Gamma}
&\geq \frac{1}{\log_{\gamma} |\O|}
		\cdot \sum_{i \in \{1, \ldots, i_{max}\} : i \text{ is of type 2}}
		\sum_{(I_A,I_B) \in \mathcal{I}_i}
		\Exp{ |P(I_B)|} \\
&\geq \frac{\lg \gamma}{\lg |\O|} \cdot \frac{i_{max}}{2} \cdot
		\frac{|\O| \cdot \lg^2 n}{8\gamma} \\
&\geq \Omega( |\O| \cdot \lg^2 n).
\end{align*}
The second inequality holds since by definition every type $2$ level $i \in \{1,
\ldots, i_{max}\}$ has $\Exp{| P(I_B) |} \geq |\O| \cdot \lg^2 n/(8
\gamma)$. The third inequality holds since $i_{max} = \Omega(\gamma
\cdot \lg |\O|/\lg \gamma)$. If other words, if
Claim~\ref{cl:new:2} were not true, then it would imply that the
data structure makes  $\Omega(|\O| \cdot \lg^2 n)$ cell-probes in
expectation while processing the input sequence $\O$. But this would
contradict the working assumption specified in the statement of
Lemma~\ref{cl:noEfficientDatastructure}.
\end{proof}

\begin{claim}
\label{cl:new:3}
The number of type $3$ levels in $\{1, \ldots, i_{max}\}$ is strictly less than $i_{max}/2$. 
\end{claim}

\begin{proof}
Let $\Gamma$ be a random variable that denotes the number of cell-probes made by the data structure while processing the (random) input sequence $\O$. Then we have:
\begin{equation}
\label{eq:sum}
\sum_{i} \sum_{(I_A, I_B) \in \I_i} C(I_A, I_B) \leq \Gamma
\end{equation}
Equation~\ref{eq:sum} holds since each cell-write made by the data structure contributes at most once to its left hand side (LHS). To see why this is true, consider the scenario where the data structure writes a cell $c$ while processing an input $x$ (say) in the input sequence $\O$. Suppose that the same cell $c$ is read by the data structure while answering a subsequent query $y$ in $\O$, and, furthermore, the cell is {\em not} read by the data structure while processing any other input that appears in the interval between $x$ and $y$. Let $I_x$ and $I_y$ respectively denote the leaf-nodes in the hierarchy tree containing $x$ and $y$, and suppose that $I_x \neq I_y$. Let $I$ be the least common ancestor of $I_x$ and $I_y$ in the hierarchy tree, and let $\DEC(I) = (I_A, I_B)$. Then the cell-write of $c$ at $x$ contributes one towards the counter $C(I_A, I_B)$, and zero towards every other counter $C(I'_A, I'_B)$. Thus, the net contribution of the cell-write of $c$ at $x$ towards the LHS is one. In contrast, if it were the case that $I_x = I_y$, or if the cell $c$ was not read at all while processing any query that appears after $x$ in $\O$, then the net contribution of the cell-write  at $x$ towards the LHS would have been zero. To summarize, we conclude that each cell-write made by the data structure contributes at most one towards the LHS.

By definition, every type $3$ level $i$ has
\begin{align*}
	\sum_{(I_A, I_B) \in \I_i} \Exp{C(I_A, I_B)} \geq \frac{1}{1200 (c + 1002)} \cdot \frac{|\O|}{\gamma} \cdot \frac{\lg n}{\lg \lg n}.
\end{align*}
Let $K_3 \subseteq \{1, \ldots, i_{max}\}$ denote the set of type $3$ levels. Now,  Equation~\ref{eq:sum} implies that:
\begin{align*}
\Exp{\Gamma}
&\geq \sum_{i} \sum_{(I_A, I_B) \in \I_i} \Exp{C(I_A, I_B)} \\
&\geq \sum_{i \in K_3} \sum_{(I_A, I_B) \in \I_i} \Exp{C(I_A, I_B)} \\
&= |K_3| \cdot \frac{1}{1200 (c + 1002)} \cdot \frac{|\O|}{\gamma} \cdot \frac{\lg n}{\lg \lg n}.
\end{align*}
Rearranging the terms in the above inequality, we get:
\begin{align*}
|K_3|
&\leq 1200(c+1002) \cdot \frac{\Exp{\Gamma} \cdot \gamma}{|\O| \cdot \frac{\lg n}{\lg \lg n}} \\
&= 1200(c+1002) \cdot \frac{o\left(|\O| \cdot \left(\frac{\lg n}{\lg \lg n}\right)^2 \right) \cdot \gamma}{|\O| \cdot \left(\frac{\lg n}{\lg \lg n} \right)} \\
&= o(1) \cdot \gamma \cdot \frac{\lg n}{\lg \lg n} \\
&< i_{max}/2.
\end{align*}
In the above derivation, the first equality holds since as per the statement of
Lemma~\ref{cl:noEfficientDatastructure}, the data structure probes $o(|\O| \cdot
		(\lg n/\lg \lg n)^2)$ cells in expectation while processing the input
sequence $\O$.  The third equality holds since $\gamma = \lg^{\beta} n$ and
$|\O| = n^{\alpha}$ for some constants $\alpha, \beta > 0$. The last inequality
holds since $i_{max} = 0.1 \gamma \log_{\gamma} |\O|$, $\gamma = \lg^{\beta} n$,
	  $|\O| = \Theta(n^{\alpha})$ and $\alpha, \beta$ are constants.
\end{proof}

\section{Proof of Theorem~\ref{Thm:DynamicOMv}}
\label{sec:omv}

Throughout this section, we will continue with the notations introduced in Section~\ref{Sec:GeneralFramework}. Further, we will set the values of the parameters $\alpha, \beta, \gamma$ and $\kappa$ as follows.
\begin{equation}
\label{eq:eqTru}
\alpha := 2, \beta := 2000, \gamma := \lg^{\beta} n = \lg^{2000} n, \text{ and } \kappa := 2/3.
\end{equation}

\subsection{Defining the random input sequence \texorpdfstring{$\O$}{O}}

\noindent
The (random) input sequence $\O$ consists of $n^2$ updates and $n$ queries. Such an input sequence is of size $\Theta(n^{\alpha})$ since we have set $\alpha = 2$ (see Equation~\ref{eq:eqTru}). The input sequence $\O$ is constructed as follows.
First, we define a sequence of $n^2$ updates: For $1 \leq  k \leq n^2$, the $k^{th}$ update is denoted by $(i_{k}, j_{k}, x_{k})$, and it consists of a {\em location} $(i_k, j_k) \in [1, n] \times [1, n]$ and a {\em value} $x_k \in \F$.  The $k^{th}$ update sets the matrix entry $M_{i_{k}j_{k}}$ to value $x_{k} \in \F$, where $x_{k}$ is picked uniformly at random from $\F$. However, the {\em location} of the $k^{th}$ update, given by $(i_k, j_k)$, is fixed deterministically.  
To finish  the construction of the sequence $\O$,  after each $n^{th}$ update we add a query chosen uniformly at random from $\F^n$. Formally,
after each update $(i_{r \cdot n}, j_{r \cdot n}, x_{r \cdot n})$ for $1 \leq r \leq n$, we insert
a uniformly random query  $v_{r} \in \F^n$.
We ensure that the sequence of locations of the updates  are \emph{well-spread}, which means that they satisfy two properties. 
\begin{enumerate}
\item All the pairs $(i_{k}, j_{k})$ are mutually disjoint.
\item For every index $n^{4/3} \leq r \leq n^2$ and every set of $n/2$ row indices $S \subseteq \{1,\dots,n\}$, there exists
    a subset $S^* \subseteq S$ of size $|S^*| \leq 8n^2 / r$ such that $|\bigcup_{k \leq r : i_{k} \in S^*} \{ j_k \} | \geq n/4$.
\end{enumerate}
Such a well-spread sequence of indices exists due to Lemma~2 in \cite{clifford2015new}.

\subsection{Proving that \texorpdfstring{$\O$}{O} is well-behaved}

\smallskip
\noindent
We begin by defining some additional  notations.
Let $u(X)$ and $q(X)$ respectively denote the number of updates and queries in an interval $X \subseteq \O$. It follows that $|X| = q(X) + u(X)$. For $1 \leq j \leq q(X)$, let $X_{j}$ denote the sequence of inputs in $X$ preceding the $j^{th}$ query in $X$. Note that $X_{j}$ is always a prefix of $X$. Consider any two consecutive intervals $X, Y \subseteq \O$. The counter $C_j(X, Y)$ denotes the number of times the following event occurs: While answering the $j^{th}$ query in $Y$, the data structure probes a cell that was written in $X$ but was not previously probed in $Y_j$. Recall the definition of the counter $C(X, Y)$ from Section~\ref{Sec:GeneralFramework}, and note that $C(X, Y) =\sum_{j=1}^{q(Y)} C_j(X, Y)$. 
Our main challenge will be to prove the lemma below. The proof of Lemma~\ref{lem:DynamicOMv} appears in Section~\ref{sec:lem:DynamicOMv}.

\begin{lemma}
\label{lem:DynamicOMv}
Every data structure for the dynamic OMv problem satisfies the following property while processing the random input sequence $\O$ described above. Consider any  pair of consecutive intervals $X, Y \subseteq \O$ with $|X| \simeq \gamma \cdot |Y|$ and $|Y| \geq |\O|^{\kappa}$, such that  $\Exp{ |P(X)| } \leq |X| \cdot \lg^2 n$ and  $\Exp{ |P(Y)| } \leq |Y| \cdot \lg^2 n$.
Then we must have:
\begin{align*}
	\Exp{C_j(X, Y)} \geq \frac{1}{100 (c+1002)} \cdot n \cdot \left(\lg n/\lg \lg n\right) \text{ for all } 1 \leq j  \leq q(Y).
\end{align*}
\end{lemma}

\noindent
{\bf Proof of Theorem~\ref{Thm:DynamicOMv}}. Consider any data structure for the dynamic OMv problem, and any pair of consecutive intervals $X, Y \subseteq \O$ with $|X| \simeq \gamma \cdot |Y|$ and $|Y| \geq |\O|^{\kappa}$. If either condition (1) or condition (2) as stated in Definition~\ref{def:well:behave} gets violated, then we have nothing more to prove. Henceforth, we assume that both the conditions (1) and (2) hold, so that we have  $\Exp{|P(X)|} \leq |X| \cdot \lg^2 n$ and $\Exp{|P(Y)|} \leq |Y| \cdot \lg^2 n$. Now, applying Lemma~\ref{lem:DynamicOMv}, we get:
\begin{equation}
\label{eq:eqBre}
\Exp{C(X, Y)}
= \sum_{j = 1}^{q(Y)} \Exp{C_j(X, Y)}
\geq \frac{1}{200 (c+1002)} \cdot q(Y) \cdot n \cdot \frac{\lg n}{\lg \lg n}.
\end{equation}
Note that $|Y| \geq |\O|^{\kappa} = \Theta(n^{\alpha \kappa}) =
\Theta(n^{4/3})$. The last equality holds since $\alpha = 2 $ and $\kappa = 2/3$
as per Equation~\ref{eq:eqTru}, and because the sum only contains a
finite number of summands. Recall that the input sequence $\O$
contains a query after every $n$ updates. Thus, the size of the interval $Y
\subseteq \O$ is large enough for us to infer that $q(Y) \geq \frac{1}{2} |Y|/n$.
Plugging this in Equation~\ref{eq:eqBre}, we get:
\begin{align*}
\Exp{C(X, Y)} \geq \frac{1}{200 (c+1002)} \cdot |Y| \cdot \frac{\lg n}{\lg \lg n}.
\end{align*}
Thus, condition (3) as stated in Definition~\ref{def:well:behave} gets violated
whenever the conditions (1) and (2) hold.  This implies that the input sequence
$\O$ is well-behaved, and concludes the proof of Theorem~\ref{Thm:DynamicOMv}.

\subsection{Proof of Lemma~\ref{lem:DynamicOMv}.}
\label{sec:lem:DynamicOMv}

Throughout the proof of Lemma~\ref{lem:DynamicOMv}, we fix the following quantities.
\begin{enumerate}
\item A data structure for the dynamic OMv problem. 
\item Two consecutive intervals $X, Y \subseteq \O$ such that $|X| \simeq \gamma \cdot |Y|$ and $|Y| \geq |\O|^{\kappa}$. 
\item An index $1 \leq j \leq q(Y)$. 
\item All the inputs in $\O$ that appear before the beginning of the interval $X$.
\item All the inputs in $\O$ that appear after the $j^{th}$ query in $Y$. 
\item All the inputs in $Y_j$. 
\end{enumerate}
Thus, everything is fixed except  the values of the updates and the queries  in $X$ and the $j^{th}$ query in $Y$.
Conditioned on these events, we next assume that:
\begin{equation}
\label{eq:assume:X}
 \Exp{|P(X)|}  \leq |X| \cdot \lg^2 n \text{ and }  \Exp{|P(Y)|}  \leq |Y| \cdot \lg^2 n 
   \end{equation}
Finally, for the sake of contradiction, we assume that:
\begin{align}
\label{eq:assume:contradict}
 \Exp{C_j(X,Y)} &< \frac{1}{100(c+1002)} \cdot n \cdot \left(\lg n/\lg \lg n\right)
\end{align}
We  now show how to encode the sequence of values of the updates in $X$ using  fewer than $u(X) \cdot \lg |\F|$ bits. This leads to a contradiction since the entropy of the object under consideration is exactly $u(X) \cdot \lg |\F|$ bits. This  concludes  the proof of Lemma~\ref{lem:DynamicOMv}.

\smallskip
\noindent {\bf Notations.}
We define some notations  that will be used in the encoding proof. We let $M_X$ and $M_{Y_j}$ respectively denote the state of the matrix $M$ just after the interval $X$ and $Y_j$. For each row $1 \leq i \leq n$, we let $m_{X, i}$ and $m_{Y_j, i}$ respectively denote the vectors in $\F^n$ that correspond to the row $i$ of matrices $M_X$ and $M_{Y_j}$. Consider any set of indices $S \subseteq \{1,\dots,n\}$ and any vector $v \in \F^n$. We let $v^{|S} \in \F^{|S|}$ denote the vector with one entry $v(i)$ for each $i \in S$.  In other words, this gives the {\em restriction} of the vector $v$ into the coordinates specified by the indices in $S$. For $1 \leq i \leq n$, we let $R_{i}$ denote the set of column indices updated in the $i^{th}$ row of $M$ during to the interval $X \subseteq \O$. More formally, we have $R_{i} = \{ j' : (i', j', \cdot) \in X \}$. 

\smallskip
\noindent {\bf Preliminaries.}
 We  introduce the concept of a {\em rank sum} in Definition~\ref{def:ranksum}.
 To get some intuition behind this definition, recall that the locations of the
 updates in $X$ are fixed and mutually disjoint. Only the {\em values} of the
 updates in $X$ can vary. Furthermore, since the updates preceding $X$ are fixed
 in advance, the values of the remaining entries in the matrix $M_X$ are known
 to the decoder.  Thus,  to encode the sequence of updates in $X$, it suffices
 to encode the vectors $m_{X,i}^{|R_i}$ for all $1 \leq i \leq n$.  Towards this
 end, the encoder will first find a suitable set of vectors $\{v_1, \ldots, v_k
 \} \subseteq \F^n$, for some positive integer $k$ whose value will be
 determined later on. Next,  as part of the procedure for encoding the vectors
 $m_{X,i}^{|R_i}$, she will  convey (in an indirect manner to be specified
		 later) to the decoder the results of the inner products $\left\langle
 m_{X, i}^{|R_i} , v_{k'}^{|R_i} \right\rangle$ for all $1 \leq k' \leq k$.  For
 this encoding to be efficient, the decoder should be able to retrieve a lot of
 information about the vectors $m_{X, i}^{|R_i}$ by looking at the results of
 these inner products. This means that  for most of the rows $i \in [1, n]$ we
 want most of the vectors $v_1^{|R_i}, \ldots, v_k^{|R_i}$ to be linearly
 independent, or, equivalently, for most of the rows $i \in [1, n]$ we want
 $\dim\left(\spane\left(v_1^{|R_i}, \ldots, v_k^{|R_i}\right)\right)$ to be
 large. This intuition can be formalized by saying that we want the {\em rank
	 sum} $\RS(v_1, \ldots, v_k)$, as defined below, to be large.

\begin{definition}
\label{def:ranksum}
	    The {\em rank sum} of a set of $k$ vectors $\{v_1, \ldots, v_k \} \subseteq \F^n$
	    is	given by $\RS(v_1, \dots, v_k) = \sum_{i=1}^n \dim\left(\spane\left(v_1^{|R_i}, \ldots, v_k^{|R_i}\right)\right)$.
\end{definition}

Recall that the encoder will have to convey the results of the inner products $\left\langle m_{X,i}^{|R_i} , v_{k'}^{|R_i} \right \rangle$ to the decoder {\em in an indirect manner}. We now elaborate on this aspect of the encoding procedure in a bit more details. Basically, the encoder will identify a small set of cells $\C^* \subseteq P(X)$, and send their addresses and contents (at the end of the interval $X$) to the decoder. This set $\C^*$ will contain the necessary information about the inner products $\left\langle m_{X,i}^{|R_i} , v_{k'}^{|R_i} \right \rangle$. To see why this is possible, suppose that the decoder simulates the data structure from the beginning of $\O$ till just before the interval $X$, then skips the intervals $X$ and $Y_j$, and then tries to simulate a query $v \in \F^n$ at the $j^{th}$ position of $Y$. Furthermore, suppose that the decoder gets {\em lucky}, meaning that the query algorithm never has to read the content of a cell in $P(X) \setminus \C^*$. In such an event we say that the set $\C^*$ ``resolves'' the vector $v$. The key insight is that if $\C^*$ resolves $v$, then the decoder  can recover the vector $M_{Y_j} \cdot v$ by looking at the contents of the cells in $\C^*$ (and a few other minor things that will be specified later on). Since the inputs in $Y_j$ are fixed in advance, from $M_{Y_j} \cdot v$ the decoder can  infer the vector $M_X \cdot v$. Similarly, since the inputs preceding $X$ are fixed in advance, from $M_X \cdot v$ the decoder can infer the inner products $\left\langle m_{X,i}^{|R_i} , v^{|R_i} \right \rangle$. To summarize, the encoder will convey to the decoder the inner products $\left\langle m_{X,i}^{|R_i} , v^{|R_i} \right \rangle$ {\em in an indirect manner}, by sending her the addresses and contents of the  cells in $\C^*$. 

Now, recall the motivation behind Definition~\ref{def:ranksum}. It implies that ideally we would like to have a small set of cells $\C^* \subseteq P(X)$ that resolves a ``nice'' set of vectors in $\F^n$ with large rank sum. Furthermore, the decoder will also need  to identify such a ``nice'' set of vectors. One way to do this is to require that there are a large number of subsets of $\F^n$ that are nice. If this is the case, then the encoder and the decoder can use shared randomness to sample some subsets of $\F^n$ uniformly at random, and with good enough probability, one of the sampled subsets will be nice. This intuition is formalized in Definition~\ref{def:resolve}, where $\M(\C)$ corresponds to the collection of all such nice subsets of query vectors $v \in \F^n$ for the set of cells $\C \subseteq P(X)$.

\begin{definition}
\label{def:resolve}
Consider any subset of cells $\C \subseteq P(X)$ and any vector $v \in \F^n$. We say that  $\C$ ``resolves''  $v$ iff the data structure does not probe any cell in $P(X) \setminus \C$ when it is asked to return the answer $M v$ by the $j^{th}$ query in $Y$. We let $Q(\C) \subseteq \F^n$ denote the set of vectors that are resolved by $\C$. Finally, we let $\M(\C) \subseteq 2^{Q(\C)}$ denote the collection of all subsets $\{v_1, \ldots, v_k\} \subseteq Q(\C)$ of size $k = (|X| - 1)/n$ with rank sum $\RS(v_1, \dots, v_k) \geq nk / 32$.
\end{definition}

\begin{lemma}
\label{lem:CellSamplingDynamicOMv}
	
	With probability at least $1/4$ over the randomness of $X$,
	there exists a subset of cells $\C^* \subseteq P(X)$ of size
	$|\C^*| = |X| \cdot \lg |\F| / (1024w)$ such that
	$|\M(\C^*)| \geq |\F|^{0.999 \cdot nk}$.
\end{lemma}

Proving Lemma~\ref{lem:CellSamplingDynamicOMv} requires a standard application
of the cell-sampling technique. We defer the proof of
Lemma~\ref{lem:CellSamplingDynamicOMv} to the end of the section (see
Section~\ref{sec:ProofCellSamplingDynamicOMv}).
Instead we focus on showing that Lemma~\ref{lem:CellSamplingDynamicOMv} leads to
an encoding of the sequence of updates in $X$, and that the resulting encoding
uses less than $u(X) \cdot \lg |\F|$ bits in expectation.  This concludes the
proof of Lemma~\ref{lem:DynamicOMv}.

\smallskip
\noindent {\bf Shared randomness between the encoder and the decoder.} 

\smallskip
\noindent We assume that both the encoder and the decoder have access to a joint source of random bits.
 They use this random source to sample $m = |\F|^{nk/512}$ sets
$\Gamma_1, \dots, \Gamma_m$. Each set $\Gamma_i$ consists of $k = (|X| - 1) / n$ vectors that are picked uniformly at random from $\F^n$.

\smallskip
\noindent {\bf Encoding the sequence of updates in $X$.}

\smallskip
\noindent
Recall that the encoder and the decoder  know  the inputs preceding $X$ and 
the inputs in $Y_j$, and both of them know the sets $\Gamma_1, \dots, \Gamma_m \subseteq \F^n$.
The encoding procedure works as follows.

\noindent {\bf Step 1.} 
			We simulate the data structure  until the end of $X$.
			If there exists a  set of cells $\C^* \subseteq P(X)$ as per Lemma~\ref{lem:CellSamplingDynamicOMv},
			then we proceed to  step (2) of the encoding.
			Otherwise we send a $0$-bit and  the naive encoding of the sequence of updates
			in $X$, and then   we terminate the encoding procedure. By Lemma~\ref{lem:CellSamplingDynamicOMv}, the probability of this event (that no such $\C^*$     exists) is at most $3/4$.

\noindent {\bf Step 2.}		
			Since we are in step (2), we must have found a set of cells $\C^* \subseteq P(X)$ as per Lemma~\ref{lem:CellSamplingDynamicOMv}.
			We check if there exists an index $1 \leq i^* \leq m$ such that $\Gamma_{i^*} \in \mathcal{M}(\C)$.
			If such an index exists, then we proceed to step (3) of the encoding.
			Otherwise we send a $0$-bit and  the naive encoding of the sequence of updates in $X$, and then we terminate the encoding procedure.
			The probability of this event (that no such index $i^*$ exists)  is
			at most
			$(1 - |\mathcal{M}(\C)| / \binom{|\F|^n}{k})^m
				\leq \exp( -m |\mathcal{M}(\C)| / |\F|^{nk}) \leq \exp( -|\F|^{\Omega(nk)}) \leq 1/100$.
				The second inequality holds since $m = |\F|^{nk/512}$ and
				$|\M(\C)| \geq |\F|^{0.999 \cdot nk}$.

\noindent {\bf Step 3.}
  			We send a $1$-bit and then send an encoding of the index $i^*$, followed by
			the addresses and contents\footnote{For every cell in $\C^*$, we encode its content at the end of the interval $X$.} of the cells in $\C^*$.  Encoding the index $i^*$ takes $\lg m$ bits. Encoding the address and content of one cell requires $2 w$ bits. Hence,
			the total number of bits sent is: $1+\lg m + (2w) \cdot |\C^*|  = 1 + \lg |\F|^{nk/512} + (2w) \cdot |X| \cdot \lg |\F|  / (1024 w)
							\leq 1 + (nk) \cdot \lg |\F|/512 + |X| \cdot \lg |\F| / 512 = 1 + (|X|-1) \cdot \lg |\F|/512 + |X| \cdot \lg |\F|/512
							\leq 1 + |X| \cdot \lg |\F|/256$.

\noindent {\bf Step 4.}	 We send the addresses and contents\footnote{For every cell in $P(Y_j)$, we encode its content at the end of the interval $Y_j$.} of the cells in  $P(Y_{j})$.
In expectation, the total number of bits required is:
$\Exp{ |P(Y_{j})| }  \cdot (2w) \leq \Exp{|P(Y)|} \cdot (2w) \leq
|Y| \cdot \lg^2 n \cdot (2w)  = o(|X|)$. The second inequality
follows from Equation~\ref{eq:assume:X}. The last equality holds
since $w = \lg n$, $|X| \simeq \gamma \cdot |Y|$ and $\gamma =
\Theta(\lg^{2000} n)$ as per Equation~\ref{eq:eqTru}.

\noindent {\bf Step 5.}
			We iterate over the rows of the matrix $M_{X}$ from $1$ to $n$.
			Let  the $k$ vectors in $\Gamma_{i^*}$ be denoted by $v_1, \dots, v_k$.
			For each row $1 \leq i \leq n$, we proceed as follows:
			\begin{itemize}
				\item[(a)]
						We create an empty set $T_i$. Next, we iterate over all the vectors in $\F^{|R_{i}|}$ in a predefined 
						and fixed order that is known to the decoder. For each vector $v \in \F^{|R_{i}|}$ in that order, if
						$v \notin \spane\left( v_1^{|R_{i}}, \dots, v_k^{|R_{i}}, T_i \right)$,
						then we add the vector $v$ to $T_i$ by setting $T_i := T_i \cup \{v\}$.
				\item[(b)]
					After finishing Step 5 (a),
						we compute  $\left\langle m_{X,i}^{|R_{i}}, v \right\rangle$ for each vector $v \in T_i$ and send this inner product.
			\end{itemize}

			\noindent Sending the inner products in Step~5 (b) requires $|T_i| \cdot \lg |\F|$ bits for a given row $i$.
			When Step~5 (a) is finished, we have
			$|T_i| = |R_{i}| - \dim\left(\spane\left( v_1^{|R_{i}}, \dots, v_k^{|R_{i}} \right)\right)$.
			Hence, the total number of bits sent for the inner products for all rows is equal to:
			\begin{align*}
				\sum_{i = 1}^n |T_i| \cdot \lg |\F|
				&= \sum_{i=1}^n \left\{ |R_{i}| - \dim\left(\spane\left( v_1^{|R_{i}}, \dots, v_k^{|R_{i}} \right)\right) \right\} \cdot \lg |\F| \\
				& =  \sum_{i=1}^n |R_i| \cdot \lg |\F| - \sum_{i=1}^n \dim\left(\spane\left( v_1^{|R_{i}}, \dots, v_k^{|R_{i}} \right)\right) \cdot \lg |\F| \\
				& =  u(X) \cdot \lg |\F| - \RS(v_1, \dots, v_k)   \cdot \lg |\F| \leq  ( |X| - nk/32) \cdot \lg |\F| \\
				& = (|X| - (|X| - 1)/32) \cdot \lg |\F|  \leq  (31/32) \cdot |X| \cdot \lg |\F|
			\end{align*}
			The third equality follows from Definition~\ref{def:ranksum} and the fact that no two updates in the input sequence $\O$ change the same entry in the matrix $M$, which implies that $\sum_{i=1}^n |R_i| = u(X)$. The first inequality follows from Definition~\ref{def:resolve} and the fact that $u(X) \leq |X|$.

\noindent This concludes the description of the encoding procedure.

\begin{claim}
\label{cl:encoding}
The encoding described above requires fewer than $u(X) \cdot \lg |\F|$ bits in expectation.
\end{claim}

\begin{proof}
Applying union bound, we get: the probability that the encoding procedure terminates in either Step 1 or 2 is at most $3/4 + 1/100 < 4/5$. If the encoding procedure terminates in either Step 1 or  2, then the total number of bits in the encoding is  (say) $\ell_1 = 1 + u(X) \cdot \lg |\F|$.

Next, we bound  the expected number of bits  (say, $\ell_2$) used by the encoding, conditioned on the event that it executes Steps 3--5. From the description of Steps~3--5, we get:
\begin{align*}
	\ell_2 	& = 1 + (1/256) \cdot |X| \cdot \lg |\F| + o(|X|) + (31/32) \cdot |X| \cdot \lg |\F|  \\
	& \leq  1 + o(|X|) + (99/100) \cdot |X| \cdot \lg |\F|  < (199/200) \cdot u(X) \cdot \lg |\F|.
\end{align*}
The last inequality holds since the input sequence $\O$ contains a query after every $n$ updates, and hence we have $u(X) \geq (1-1/n) \cdot |X| > (999/1000) \cdot |X|$ for large enough $|X|$.
To summarize, the expected number of bits required by the encoding is
at most $(4/5) \cdot \ell_1 + (1/5) \cdot \ell_2 <  u(X) \cdot \lg |\F|$. This concludes the proof of the claim.
\end{proof}

\noindent {\bf Decoding the sequence of updates in $X$.}

\smallskip
\noindent
Since the matrix entries for the updates are fixed in advance and since these entries are mutually disjoint, one can reconstruct
the sequence of updates in $X$ if one gets to know the vector $m_{X,i}^{|R_i}$ for each row $1 \leq i \leq n$. Accordingly, the goal of the decoding procedure will be to recover these vectors $m_{X, i}^{|R_i}$. It will consist of the following steps.

\noindent {\bf Step 1.} If the first bit of the encoding is a $0$, then we just restore all updates from
			the naive encoding and terminate the procedure.

\noindent {\bf Step 2.} If the first bit of the encoding is a $1$, then we
			recover the index $1 \leq i^* \leq m$ and the set of cells $\C^* \subseteq P(X)$. Let the $k$ vectors in the set $\Gamma_{i^*} \subseteq \F^n$ be denoted by
			$v_1, \dots, v_k$.
			
\noindent {\bf Step 3.} We now reconstruct the sets $T_i$.
			For each row $1 \leq i \leq n$, we create an empty set $T_i$ and
			iterate over the vectors in $\F^{|R_{i}|}$ in the predefined order
		       (see Step 5(a) of the encoding).
			While considering  a vector $v \in \F^{|R_{i}|}$ during any such iteration, if we find that
			$v \notin \spane\left( v_1^{|R_{i}}, \dots, v_k^{|R_{i}}, T_i \right)$,
			then we add $v$  to  $T_i$ by setting $T_i := T_i \cup \{v\}$. 

\noindent {\bf Step 4.} We simulate the data structure through all the inputs preceding the interval $X$. Let $\mathcal{S}$ denote the collective state of the memory cells at the end of this simulation.

\noindent {\bf Step 5.} For all $v \in \{ v_1, \ldots, v_k\}$, we recover the vector $M_{Y_j} \cdot v$ for the decoder. This is done as follows.
 We ask the data structure to answer the query  $v$. 	We  allow the data structure to write any cell while answering the query.
 In contrast, whenever the data structure tries to {\em read} a cell $c$ (say) while answering the query,  we perform the following operations.
		\begin{itemize}
		\item[(a)] If $c \in P(Y_j)$, then we fetch the content of  $c$ from Step 4 of the encoding. 
		\item[(b)] Else if  $c \in \C^* \setminus P(Y_j)$, then we fetch the content of $c$ from Step 3 of the encoding. 
		\item[(c)] Else if $c \notin \C^* \cup P(Y_j)$, then  we claim  $c \notin P(X)$. 
		To see why the claim holds, note that $\Gamma_{i^*} \in \M(\C^*)$, and hence,  by Definition~\ref{def:resolve} we have $\Gamma_{i^*} \subseteq Q(\C^*)$. Since $v \in \Gamma_{i^*}$, we infer that $v \in Q(\C^*)$. In other words, the set of cells $\C^*$ resolves the query vector $v$.  Thus, by definition, the data structure does not probe any cell in $P(X) \setminus \C^*$ while answering the query vector $v$ at the $j^{th}$ query position in $Y$.  Since $\C^* \subseteq P(X)$, and since we are considering a scenario where $c \notin \C^* \cup P(Y_j)$, it follows that $c \notin P(X)$.  So we fetch the content of  $c$ from the state $\mathcal{S}$ of the main memory, as defined in Step (4) of the decoding.  
	\end{itemize}
	This shows that we can recover the vectors $M_{Y_j} \cdot v$ for all $v \in \{ v_1, \ldots, v_k\}$.

\noindent {\bf Step 6.} For all $v \in \{ v_1, \ldots, v_k\}$, we now recover the vector $M_{X} \cdot v$ for the decoder. Let $M''$ be an $n \times n$ matrix over $\F$ that is defined as follows. An entry $(x, y) \in [1, n] \times [1, n]$ in this matrix is set to zero if it is {\em not} updated during the interval $Y_j$; otherwise it is set to the value of the entry $(x, y)$ in $M_{Y_j}$. Note that $M_{X} = M_{Y_j} - M''$, and, furthermore, the matrix $M''$ is known to the decoder since we have fixed the inputs in $Y_j$. Accordingly, for all $v \in \{v_1, \ldots, v_k\}$, the decoder computes $M_{X} \cdot v$ from the equation: $M_X \cdot v = M_{Y_j} \cdot v - M'' \cdot v$. This is feasible since the decoder already knows the vector $M_{Y_j} \cdot v$ from step (5) above.  

\noindent {\bf Step 7.}
			For each row $1 \leq i \leq n$, we now recover the vector $m_{X, i}^{|R_i}$ as follows.
		
\noindent  From step (6) of the decoding, the decoder knows the inner product $\left\langle m_{X, i} , v\right\rangle$ for all $v \in \{v_1, \ldots, v_k\}$. Since the decoder also knows the updates in the matrix preceding $X$, from $\left\langle m_{X, i} , v\right\rangle$ she can easily infer the inner product $\left\langle m_{X, i}^{|R_i}, v^{|R_i}\right\rangle$ for every $v \in \{ v_1, \ldots, v_k\}$.  Additionally, from step (5) of the encoding and step (3) of the decoding, the decoder knows every vector $v \in T_i$ and the corresponding  inner product $\left\langle m_{X, i}^{|R_i}, v\right\rangle$. As $\dim \left(\spane \left(v_1, \ldots, v_k,  T_i\right)\right) = |R_i|$, the decoder can recover the vector $m_{X, i}^{|R_i}$ from all these $k+|T_i|$ inner products.

			This concludes the description of the decoding procedure, and  the proof of Lemma~\ref{lem:DynamicOMv}.

\subsection{Proof of Lemma~\ref{lem:CellSamplingDynamicOMv}}
\label{sec:ProofCellSamplingDynamicOMv}

We begin by defining a new counter $C_j(X, Y, v)$. Consider a scenario where
the data structure is asked to answer a query $v' \in \F^n$ in the $j^{th}$
query position of the interval $Y$. The counter $C_j(X, Y, v')$  keeps track of
the number of times the following event occurs: While answering the query $v'$,
the data structure probes a cell that was written during the interval $X$ but
was {\em not} read during the interval $Y_j$.  Recall the definition of the
counter $C_j(X, Y)$ from Section~\ref{Sec:GeneralFramework}, and note that
$C_j(X, Y) = \Exp[v]{C_j(X, Y, v)}$.  Since  we have fixed the input sequence
in $Y_j$ and the input sequence preceding the interval $X$, the counters
$C_j(X, Y)$ and $C_j(X, Y, v)$  are completely determined  by  the sequence of
values of the updates in $X$. For simplicity, henceforth we omit $Y$ from these
notations and instead write them as $C_j(X)$ and $C_j(X, v)$

\begin{claim}
\label{app:cl:good}
	Let $\mathcal{E}$ denote the event where
	$|P(X)| \leq 16 \cdot |X| \cdot \lg^2 n$
	and $C_j(X) < \frac{16}{100(c+1002)} \cdot n \cdot \frac{\lg n}{\lg \lg n}$.
	Then we have $\Pr[\mathcal{E}] \geq 1/4$.
\end{claim}

\begin{proof} Let $\mathcal{E}_1$ denote the event that $|P(X)| > 16 \cdot |X|
\cdot \lg^2 n$.  Markov's inequality and Equation~\ref{eq:assume:X} imply that
$\Pr[\mathcal{E}_1] \leq 1/16$. Let $\mathcal{E}_2$ denote the event that
$C_j(X) \geq \frac{16}{100(c+1002)} \cdot n \cdot (\lg n/\lg \lg n)$. Markov's inequality and
Equation~\ref{eq:assume:contradict} imply that $\Pr[\mathcal{E}_2] \leq 1/16$.
Since $\mathcal{E}^c = \mathcal{E}_1 \cup \mathcal{E}_2$, applying a union
bound we get: $\Pr[\mathcal{E}] \geq 1 - \Pr[\mathcal{E}_1] -
\Pr[\mathcal{E}_2] \geq 1 - 1/16 - 1/16 \geq  1/4$.  \end{proof}

For the rest of this section, we fix any input sequence $X$ that might occur
under the event $\mathcal{E}$, and then prove the existence of a set of cells
$\C^* \subseteq P(X)$ of size $\Delta = |X| \cdot \lg |\F|/(1024 w)$  such that
$|\M(\C^*)| \geq |\F|^{0.999 \cdot nk}$. This, along with Claim~\ref{app:cl:good},
implies Lemma~\ref{lem:CellSamplingDynamicOMv}.

\begin{claim} \label{app:cl:lower} Let $V(X) \subseteq \F^n$ be the  set of
vectors $v \in \F^n$ such that
$C_j(X, v) < \frac{16}{100(c+1002)} \cdot n \cdot \frac{\lg n}{\lg \lg n}$.
Then we have $|V(X)| \geq |\F^n|/4$.  \end{claim}

\begin{proof} Recall that we have conditioned on the event $\mathcal{E}$.
Hence, Claim~\ref{app:cl:good} and Markov's inequality imply that  some
constant  (say $(1/4)^{th}$) fraction of the vectors $v \in \F^n$ must have
$C_j(X, v) < \frac{16}{100(c+1002)} \cdot n \cdot \lg n/\lg \lg n$.
\end{proof}

\begin{claim} \label{app:cl:sample} Consider any query vector $v \in V(X)$.
Pick a subset of cells $\C \subseteq P(X)$ of size $|\C| = \Delta$ uniformly at
random. The probability that $\C$ resolves the query $v$ is at least
$|\F|^{-0.0001 \cdot n}$.  \end{claim}

\begin{proof} Let $P_v(X) \subseteq P(X)$ be the set of cells in $P(X)$ that
the data structure needs to probe while answering the query $v$ in the $j^{th}$
position of the interval $Y$.  Since $v \in V(X)$, we  have $C_j(X, v)  =
\lambda$ (say) where $\lambda = \frac{16}{100(c+1002)} \cdot n \cdot \lg n/\lg \lg n$. This implies that
$|P_v(X)| \leq C_j(X, v) \leq \lambda$.

The random subset of cells $\C \subseteq P(X)$ resolves the query vector $v$
iff $P_v(X) \subseteq \C$.  Hence, we can reformulate the question as follows.
We are given two sets of cells $P_v(X)$ and $P(X)$ with $P_v(X) \subseteq
P(X)$, $|P_v(X)| \leq \lambda$ and $|P(X)| \leq |X| \cdot \lg^2 n$. Now, if we
pick a subset of cells $\C \subseteq P(X)$ of size $|\C| = \Delta$ uniformly at
random from $P(X)$, then what is the probability that $P_v(X) \subseteq \C$?
Let the desired probability be $p_v(X)$. A moment's thought will reveal that:
\begin{eqnarray}
p_v(X)
& \geq &  \left(\frac{\Delta}{|P(X)|}\right) \cdot
\left(\frac{\Delta - 1}{|P(X)| -1}\right) \cdots \left(\frac{\Delta - |P_v(X)|
+ 1}{|P(X)| - |P_v(X)| + 1}  \right) \nonumber \\ 
& \geq &
\left(\frac{\Delta}{|P(X)|}\right)^{|P_v(X)|} 
\geq \left(\frac{\Delta}{|P(X)|}\right)^{\lambda}
\geq \left(\frac{\Delta}{16 \cdot |X| \cdot \lg^2 n}\right)^{\lambda} \nonumber \\ 
& = & \left(\frac{|X| \cdot \lg |\F|}{1024 \cdot w \cdot 16 \cdot |X| \cdot \lg^2 n}\right)^{\lambda}\nonumber \\ 
& \geq & \left(\frac{1}{\lg^2 n}\right)^{\lambda} \nonumber \\
& = & \left(\frac{1}{2^{2 \lg \lg n}}\right)^{\frac{16}{100(c+1002)} \cdot n \cdot \lg |\F|/\lg \lg n} \nonumber \\ 
& \geq &  |\F|^{- 0.0001 \cdot n}. \label{app:eq:prob} \end{eqnarray}
\end{proof}
 
 \begin{corollary} \label{app:cor:sample} There exists a subset of cells $\C^*
\subseteq P(X)$ of size $|\C^*| = \Delta$ such that the number of queries in
$V(X)$ resolved by $v$ satisfies the following guarantee:  $|Q(\C^*)| \geq
|\F|^{0.999 \cdot n}$.  \end{corollary}
 
 \begin{proof} Pick a subset of cells $\C \subseteq P(X)$ of size $|\C| =
\Delta$ uniformly at random from $P(X)$.  From
Claims~\ref{app:cl:lower},~\ref{app:cl:sample} and linearity of expectation, it
follows that the expected number of queries in $V(X)$ that are resolved by the
set $\C$ is at least $|V(X)| \cdot |\F|^{-o(n)} \geq \left(|\F|^n/4\right)
\cdot |\F|^{-0.0001 \cdot n} =  |\F|^{0.9999 \cdot n} / 4 \geq |\F|^{0.999 \cdot n}$. 
Hence, there must exist some such set
$\C^* \subseteq P(X)$ which resolves at least $|\F|^{0.999 \cdot n}$ many queries
in $V(X)$. It follows that we must have $|Q(\C^*)| \geq |\F|^{0.999 \cdot n}$ for
some subset of cells $\C^* \subseteq P(X)$ of size $|\C^*| = \Delta$.
\end{proof}
 
 \begin{lemma} \label{app:lm:low:rank:sum} There are at most $|\F|^{(27/32)nk}$
sets of vectors  $\{v_1, \ldots, v_k\} \subseteq \F^n$ s.t.\ $\RS(v_1,
\dots, v_k) < nk / 32$.  \end{lemma}

\begin{proof}(Sketch) Since the input sequence $\O$ is {\em well-spread}, $|X|
\geq n^{4/3}$ and $k = (|X| - 1)/n$, we can apply Lemma~5 from Clifford et
al.~\cite{clifford2015new}.  \end{proof}
 
 Corollary~\ref{app:cor:sample} and Lemma~\ref{app:lm:low:rank:sum} imply that
nearly all subsets of $k$ vectors in $Q(\C^*)$ have high rank sum, and,
furthermore, we have $|Q(\C^*)| \geq |\F|^{0.999 \cdot n}$. Hence, we infer that
$|\M(\C^*)| \geq |\F|^{0.999 \cdot nk}$.  This concludes the proof of
Lemma~\ref{lem:CellSamplingDynamicOMv}.

\section{Proof of Theorem~\ref{thm:main:poly}}
\label{app:Sec:PolynomialEvaluation}

In this section, we prove Theorem~\ref{thm:main:poly} and we will
continue using the notations introduced in Section~\ref{Sec:GeneralFramework}.
Further, we set the values of the parameters $\alpha, \beta, \gamma$ and $\kappa$ as follows.
\begin{equation}
\label{app:eq:eqHil}
\alpha := 1, \beta := 2000, \gamma := \lg^{\beta} n = \lg^{2000} n, \text{ and } \kappa := 1/2.
\end{equation}

\subsection{Defining the random input sequence \texorpdfstring{$\O$}{O}}
\label{app:Sec:PolyEval:HardInstance}
The (random) input sequence $\O$ consisting of $n$ updates and $n$ queries is defined as follows.
Initially, the polynomial is the zero polynomial, i.e., all roots are set to zero.
For $i = 0, \dots, n-1$, operation $2i + 1$ is an update setting the $i^{th}$ root of the polynomial to an element from
$\F$ that is picked uniformly at random,
and the $2(i+1)^{st}$ operation is a query that is picked uniformly at random from $\F$.

\subsection{Proving that \texorpdfstring{$\O$}{O} is well-behaved}
\label{app:Sec:PolyEval:WellBehaved}
We start by defining some notations. Let $u(X)$ and $q(X)$ respectively denote the number of updates and queries in an interval $X \subseteq \O$.
For $1 \leq j \leq q(X)$, let $X_j$ be the sequence of inputs in $X$ preceding the $j^{th}$ query in $X$.
Now consider two consecutive intervals $X, Y \subseteq \O$.
The counter $C_j(X,Y)$ denotes the number of times the follow event occurs:
While answering the $j^{th}$ query in $Y$, the data structures probes a cell that was
last written in $X$ but not yet accessed in $Y_j$.
Recall the definition of the counter $C(X, Y)$ from Section~\ref{Sec:GeneralFramework},
and note that $C(X, Y) =\sum_{j=1}^{q(Y)} C_j(X, Y)$. 

\begin{lemma}
\label{app:lem:PolyEval}
	Every data structure for the dynamic polynomial evaluation problem satisfies the following property
	while processing the random input sequence $\O$ described in Section~\ref{app:Sec:PolyEval:HardInstance}.
	Fix any pair of consecutive intervals $X, Y \subseteq \O$, where $|X| \simeq \gamma \cdot |Y|$
 	and $|Y| \geq |\O|^\kappa$ such that $\Exp{ |P(X)| } \leq |X| \lg^2 n$, and
	$\Exp{ |P(Y)| } \leq |Y| \lg^2 n$.
	Then we must have:
	\begin{align*}
		\Exp{C_j(X,Y)} \geq \frac{1}{200 (c+1002)} \cdot \frac{\lg n}{\lg \lg n} \text{ for all } 1 \leq j \leq q(Y).
	\end{align*}
\end{lemma}
\begin{corollary}
	The input sequence $\O$ defined in Section~\ref{app:Sec:PolyEval:HardInstance} is well-behaved
	as per Definition~\ref{def:well:behave}.
\end{corollary}
\begin{proof}
If either condition (1) or condition (2) as stated in Definition~\ref{def:well:behave} gets violated, then we have nothing more to prove.
Henceforth, we assume that both the conditions (1) and (2) hold.
Applying Lemma~\ref{app:lem:PolyEval}, we get:
\begin{equation}
\label{app:eq:eqRem}
\Exp{C(X, Y)}
= \sum_{j = 1}^{q(Y)} \Exp{C_j(X, Y)}
\geq \frac{1}{100(c+1002)} \cdot q(Y) \cdot \frac{\lg n}{\lg \lg n}.
\end{equation}
Note that $|Y| \geq |\O|^{\kappa} = \Theta(n^{\alpha \kappa}) = \Theta(\sqrt{n})$.
The last inequality holds since $\alpha = 1$ and $\kappa = 1/2$ as per
Equation~\ref{app:eq:eqHil}.
Recall that the input sequence $\O$ contains a query after each update.
Thus, the size of the interval $Y \subseteq \O$ is large enough for us to infer that
$q(Y) = |Y| / 2$. Plugging this in Equation~\ref{app:eq:eqRem}, we get:
\begin{equation}
\Exp{C(X, Y)} \geq \frac{1}{200 (c+1002)} \cdot |Y| \cdot  \frac{\lg n}{\lg \lg n}.
\end{equation}
To summarize, the condition (3) as stated in Definition~\ref{def:well:behave} gets violated whenever
the conditions (1) and (2) hold.  This implies that the input sequence $\O$ is well-behaved.
\end{proof}

\subsection{Proof of Lemma~\ref{app:lem:PolyEval}}
\label{app:Sec:lem:PolyEval}
We  prove Lemma~\ref{app:lem:PolyEval} by contradiction. Towards this end, throughout Section~\ref{app:Sec:lem:PolyEval}, we fix:
\begin{enumerate}
\item A data structure for the dynamic polynomial evaluation problem. 
\item Two consecutive intervals $X, Y \subseteq \O$ such that $|X| \simeq \gamma \cdot |Y|$ and $|Y| \geq |\O|^{\kappa}$. 
\item An index $1 \leq j \leq q(Y)$. 
\item All the inputs in $\O$ that appear before the beginning of the interval $X$.
\item All the inputs in $\O$ that appear after the $j^{th}$ query in $Y$. 
\item All the inputs in $Y_j$. 
\end{enumerate}
To summarize, only the inputs in $X$ and the $j^{th}$ query in $Y$ are allowed to vary. Everything else is fixed. 
Conditioned on these events, we next assume that:
\begin{align}
 \Exp{|P(X)|} & \leq |X| \cdot \lg^2 n \label{app:eq:assume:X} \\
 \Exp{|P(Y)|} & \leq |Y| \cdot \lg^2 n \label{app:eq:assume:Y}
  \end{align}
Finally, for the sake of contradiction, we assume that:
\begin{align}
 \Exp{C_j(X,Y)} < \frac{1}{200(c+1002)}\cdot \frac{\lg n}{\lg \lg n}.
\end{align}
Having made these assumptions, we  now show how to encode the sequence of updates in $X$ using
less than $u(X) \cdot \lg |\F|$ bits. This leads to a contradiction since the entropy of the
sequence of updates in $X$ is exactly $u(X) \cdot \lg |\F|$ bits. This  concludes the proof of Lemma~\ref{app:lem:PolyEval}.

Our proof is based on a cell-sampling argument as introduced by Larsen~\cite{larsen2012higher}.
The following lemma identifies the set of cells we will use in our encoding.
In the lemma we assume that $\F$ is partitioned into $\ell = |\F|^{1/4}$ consecutive
subsets of $|\F|^{3/4}$ elements each. We denote these subsets by $\F_1, \dots, \F_\ell$.
We say that a set of cells $\C \subset P(X)$ \emph{resolves} a query $x$, if the data structure
does not need to probe any cell from $P(X) \setminus \C$ in order to answer query $x$.
\begin{lemma}
\label{app:lem:goodCells}
	Assume $\Exp{ C_j(X,Y) } < \frac{1}{200(c+1002)} \cdot \lg n / \lg \lg n$.
	Then there is an index $k^* = k^*(X, Y)$ such that with probability $p$ at least $1/4$
	over the randomness of the updates and queries in $X$ and $Y_j$,
	there exists a set of cells $\C$, such that
	\begin{itemize}
		\item $|\C| = \frac{ |P(X)| }{ 24 \lg^2 n}$, and
		\item $\C$ resolves at least $n + 1$ queries from the set $\F_{k^*}$.
	\end{itemize}
\end{lemma}

We defer the proof of Lemma~\ref{app:lem:goodCells} to Section~\ref{app:Sec:lem:goodCells}.
Note that in the lemma the choice of the index $k^*$ depends on the
randomness of $X$ and $Y$. Hence, the encoder will need to encode it, but
this requires only $\frac{3}{4} \lg n$ bits.
We now show how we can use the cells $\C$ from the lemma to obtain an efficient
encoding of the updates and queries in the interval.

\smallskip
\noindent {\bf Encoding the sequence of updates in $X$.}
\smallskip

In the encoding we consider two cases distinguishing whether
the claimed set of cells from Lemma~\ref{app:lem:goodCells} exists or not.
The encoder can check if such a set of cells exists by enumerating all sets of cells of size $|P(X)| / (b \lg^2 n)$ and
then verifying if one of them resolves at least $n + 1$ queries from $\F_{k^*}$.

\emph{Case 1:}
There is no set of cells $\C$ with the properties from Lemma~\ref{app:lem:goodCells}.
In this scenario, the first bit of the encoding is a $0$.
After that the encoder writes down all updates in $X$ using the naive encoding.
This takes $1 + u(X) \cdot \lg |\F|$ bits.

\emph{Case 2:}
The encoder finds a set of cells $\C$ with the properties from Lemma~\ref{app:lem:goodCells}.
The encoder starts by writing a $1$ bit followed by the encoding of $k^*$.
Then it encodes the addresses and contents of the cells from $\C$ using
$2 |P(X)| / (b \lg n)$ bits.
To identify which queries from $\F_{k^*}$ can be resolved without probing any cells from
$P(X) \setminus \C$, we encode $u(X) + 1$ of the queries from $\F_{k^*}$ which
do not coincide with any of the roots of the polynomial that were set outside of $X$
(this can be done since $\C$ resolves
at least $n+1$ queries and there are $n - u(X)$ roots set outside of $X$).
This requires $\lg\left( \binom{|\F|^{3/4}}{u(X)+1} \right)$ bits.
We additionally encode the permutation to restore the order of the updates in $X$ using $\lg( u(X)! )$ bits.
Finally, we encode the updates and queries in $Y_j$ and the set $P(Y_j)$ using the naive encoding and
spending $o( u(X) )$ bits.

\smallskip
\noindent {\bf Decoding the sequence of updates in $X$.}
\smallskip

The decoding procedure works as follows.
\begin{enumerate}
\item If the encoded message starts with a $0$, then the decoder can trivially
		recover all updates from $X$ and stops the computation.

\item If the encoded message starts with a $1$, then the decoder first performs all updates and queries of $\O$
	that occured before $X$ on the data structure.
	The decoder further recovers the addresses and the contents
	of the cells in $\C$ and then those in $P(Y_j)$.
\item Next, the decoder evaluates the polynomial at the positions given by
		the $u(X) + 1$ queries that are resolved by $\C$ and sent by the encoder.
		Denote this set of queries $Q$.

	   	For $v \in Q$, the decoder proceeds as follows:
	   	The decoder runs the query procedure of the data structure for $v$.
		Whenever the data structure wants to probe the cell, it first checks if the cell is in $P(Y_j)$.
		If this is the case, it uses the information from $P(Y_j)$, otherwise, it checks if the cell is in $\C$.
		If this is the case, it uses the contents of the cell in $\C$. Otherwise, the decoder uses the cell from the
		its memory.

		Note that by choice of $Q$ and $\C$, any cell that is read from the memory cannot have been
		changed during the updates and queries in $X$.
\item After all queries were evaluated, the decoder can restore the values of the roots that were set during the updates
	in $X$ by evaluating an equation system with $u(X) + 1$ equations and $u(X)$ variables.
\item The decoder recovers the order of the updates by the permutation that was encoded.
\end{enumerate}

\smallskip
\noindent {\bf Length of the encoding.}
\smallskip

In case 1, the encoding has an expected length of $\ell_1 = u(X) \cdot \lg |\F| + 1$ bits.
In case 2, first observe that $|X| \leq 3 u(X)$ by definition of $\O$.
Then the expected length of the encoding is given by
\begin{align*}
	\ell_2 =&\frac{2 \Exp{|P(X)|}}{b \lg n} + \lg\left( \binom{ |\F|^{3/4}}{ u(X)+1 } \right) + \lg( u(X)! ) + o( u(X)) \\
	\leq& \frac{2}{b} |X| \lg n + (u(X)+1) \lg( |\F|^{3/4} / u(X) ) + u(X) \lg( u(X)) + o(u(X)) \\
	\leq& u(X) \left( 6/b \lg n + \lg(|\F|^{3/4}) \right) + o(u(X)) \\
	\leq& \left( \frac{1}{4} + \frac{3 (1+\varepsilon)}{4} \right) u(X) \lg n + o(u(X)) \\
	<& (1+\varepsilon) u(X) \lg n = u(X) \lg |\F|.
\end{align*}
Note that the encoding for case 2 takes fewer bits than given by the entropy of
the updates in $X$. Thus, the expected size of the encoding takes $(1-p) \ell_1 + p \ell_2$ bits.
As $p \geq 1/4$, this is smaller than the entropy of the updates in $X$, which is given by $u(X) \lg |\F|$ --- a contradiction.

\subsection{Proof of Lemma~\ref{app:lem:goodCells}}
\label{app:Sec:lem:goodCells}
We first prove that an index $k^*$ with the desired properties exists.
We show this by invoking the probabilistic method.
Let $C_j(X,Y,x)$ be $C_j(X,Y)$ under the assumption that the $j^{th}$ query in $Y$
is for element $x \in \F$.
Let $t_j^k = \sum_{x \in \F_k} \Exp{C_j(X,Y,x)} / |\F_k|$, i.e.,
$t_j^k$ denotes the average number of cells probed by queries from $\F_k$.
Note that $\sum_{k = 1}^\ell \Exp{t_j^k} / |\F^{1/4}| = \Exp{C_j(X,Y)} < \lg n / (K \lg \lg n)$
for $K = 80000$.
By Markov's inequality with probability at most $1/2$
over the randomness of $X$ and $Y$, $\Exp{C_j(X,Y,x)} \geq 4 \lg n / (K \lg \lg n)$.
Hence, there must be an index $k$
such that $\Exp{t_j^{k}} \leq 4 \lg n / (K \lg \lg n)$.
Thus, by the probabilistic method an index $k^*$ with the desired propertes 
must exist.

It is left to show that for $k^*$ the set $\C$ exists with probability at least $1/2$ over the randomness
of $X$. Using Markov's inequality we get that with probability at least $1/2$,
$t_j^{k^*} \leq 100 \Exp{t_j^{k^*}} \leq 400 \lg n / (K \lg \lg n)$ and
$|P(X)| \leq \Exp{|P(X)|} \leq 100 |X| \lg^2 n / 24$.
We show that if this event occurs, then $\C$ exists.
Let $G_{k^*}(X)$ denote the set of all queries from $\F_{k^*}$
which probe at most $400 \lg n / (K \lg \lg n)$ cells from $P(X)$.
Observe that $|G_{k^*}(X)| = \Omega(|\F|^{3/4})$.

Let $\Delta = |P(X)| / (24 \lg^2 n) = 100 |X| / 24$ and
consider all $\Delta$-subsets of cells from $P(X)$.
Any query in $G_{k^*}(X)$ probes at most $\mu = 400 \lg n / (K \lg \lg n)$ cells from $P(X)$.
Then there must exist a set $\C$ of $\Delta$ cells which resolves at least
$|G_{k^*}(X)| \binom{|P(X)| - \mu}{\Delta - \mu}/\binom{|P(X)|}{\Delta}$ queries:
\begin{align*}
	|G_{k^*}(X)| \frac{ \binom{|P(X)| - \mu}{\Delta - \mu} }{ \binom{|P(X)|}{\Delta} }
	&= |G_{k^*}(X)| \frac{ (|P(X)| - \mu)! \Delta! }{ |P(X)|! (\Delta - \mu)! } \\
	&\geq |G_{k^*}(X)| \cdot \left( \frac{ \Delta - \mu }{ |P(X)| } \right)^\mu \\
	&\geq |G_{k^*}(X)| \cdot \left( \frac{ 50 |X| / 24 }{ 100 |X| \lg^2 n / 24 } \right)^\mu \\
	&= |G_{k^*}(X)| \cdot \left( \frac{ 1 }{ 2 \lg^2 n } \right)^{ 400 \lg n / (K \lg \lg n) } \\
	&= |G_{k^*}(X)| \cdot 2^{-(800 / K) \lg n} \\
	&= \Omega( |\F|^{3/4} ) \cdot n^{-800 / K} \\
	&\geq |\F|^{74/100 - o(1)} \geq n+1,
\end{align*}
where we used that for large enough $n$, we have that $\Delta - \mu \gg 0.5 \Delta$.

\bibliographystyle{plain}
\bibliography{main}

\end{document}